\documentclass{article}
\usepackage{amsmath,amsthm,amssymb}
\usepackage{array, longtable}
\newtheorem{Theorem}{Theorem}[section]
\newtheorem{lem}[Theorem]{Lemma}
\newtheorem{Remark}[Theorem]{Remark}
\newtheorem{Definition}[Theorem]{Definition}

\newtheorem{Example}[Theorem]{Example}
\numberwithin{equation}{section}
\usepackage[latin1]{inputenc}

\begin{document}
\title{{\LARGE Enumeration of extended irreducible binary Goppa codes}}

\author{Bocong Chen$^1$ and Guanghui~Zhang$^2$\footnote{E-mail addresses: bocongchen@foxmail.com(B. Chen); zghui2012@126.com (G. Zhang)}} 


\date{\small
$1.$ School of Mathematics, South China University of Technology, Guangzhou 510641, China\\
$2.$ School of Mathematical Sciences, Luoyang Normal University, Luoyang, Henan 471934, China
}


\maketitle

\begin{abstract}
The family of Goppa codes   is one of the most interesting subclasses of linear codes.
As  the McEliece cryptosystem often chooses a random Goppa code   as its key,
knowledge of the number of inequivalent Goppa codes for fixed parameters may facilitate in the evaluation
of the security of such a cryptosystem. In this paper we present a new approach to give an upper  bound on the number of
inequivalent extended irreducible binary Goppa codes. To be more specific,
let  $n>3$ be an odd prime number and
$q=2^n$; let $r\geq3$ be a positive integer satisfying $\gcd(r,n)=1$ and $\gcd\big(r,q(q^2-1)\big)=1$.
We obtain an upper bound for the number of inequivalent extended irreducible binary Goppa codes of length $q+1$ and degree $r$.



\end{abstract}

\section{Introduction}
The family of Goppa codes is one of the most interesting subclasses of linear codes, for example, see \cite{berger3,berger1,goppa1,goppa2}.
This family of codes also contains long codes that have good parameters.
Goppa codes have attracted considerable attention in cryptography; the McEliece cryptosystem
and the Niederreiter cryptosystem are examples of public-key cryptosystems that use  Goppa codes, for example, see \cite{ls}, \cite{mceliece}.
One of the   reasons why Goppa codes receive interest from cryptographers may be
that  Goppa codes have few invariants and the number of inequivalent codes grow exponentially
with the length and dimension of the code, which makes it possible to resist to any structural attack.
More specifically, in the McEliece cryptosystem a random Goppa code is often chosen as a key.
When we give the assessment of the security of this cryptosystem against the enumerative attack,
it is important for us to know the number of Goppa codes for any given set of parameters.
An enumerative attack in the McEliece cryptosystem is to find all Goppa codes for a given set of parameters and
to test their equivalences with the public codes \cite{ls}.
Thus one of the key issues for the McEliece cryptosystem is the enumeration of inequivalent Goppa codes for a given set of parameters.

There has been tremendous interest in developing the enumeration of inequivalent Goppa codes.
Making use of the invariant property under the group of transformations, Moreno \cite{moreno} classified cubic and quartic irreducible
Goppa codes and obtained that there are four inequivalent  quartic Goppa codes of length $33$;
in addition, Moreno proved that there is only one inequivalent extended irreducible binary Goppa code with any length and degree $3$.
Ryan  studied irreducible Goppa codes \cite{ryan14,ryan2}.
Ryan and Fitzpatrick \cite{rf} obtained an upper bound for the number of irreducible Goppa codes of length $q^n$ over $\mathbb{F}_q$.
Ryan \cite{ryan15} made
a great improvement on  giving a much tighter upper bound (compared to the previous work \cite{ryan14})
on the number of  inequivalent extended irreducible binary quartic Goppa codes of
length $2^n+1$, where $n>3$ is a prime number. Following that line of research,
an upper bound on the number of inequivalent extended irreducible
binary Goppa codes of degree $2^m$ and length $2^n+1$ was given in \cite{mmr}, where $n$ is an odd prime and $m>1$ is a positive integer.
Magamba and Ryan \cite{magamba} obtained an upper bound on the number of inequivalent extended irreducible $q$-ary Goppa codes
of degree $r$ and length $q^n+1$, where $q=p^t$, $n$ and $r>2$ are both prime numbers.
In \cite{musukwa}, an upper bound on the number of inequivalent extended irreducible binary Goppa codes of degree $2p$ and length
$2^n+1$ was produced, where $n$ and $p$ are two distinct odd primes such that $p$ does not divide $2^n\pm1$.
Very recently, Huang and Yue \cite{yueqin} obtained an upper bound on the number of extended irreducible binary Goppa codes of degree $6$ and length
$2^n+1$, where $n>3$ is a prime number.

In this paper, we explore further the ideas in \cite{yueqin} and \cite{ryan15} to give an
upper bound on the number of inequivalent extended irreducible binary Goppa codes of length
$2^n+1$ and degree $r$, where $n>3$ is an odd prime number and
$r\geq3$ be a positive integer satisfying $\gcd(r,n)=1$ and $\gcd\big(r,q(q^2-1)\big)=1$.
By virtue of a result in \cite{ryan15}, the number of inequivalent extended irreducible binary Goppa codes of length $2^n+1$
and degree $r$ is less than or equal to the number of orbits of ${\rm P\Gamma L}$ on $\mathcal{S}$,
where ${\rm P\Gamma L}={\rm PGL}\rtimes {\rm Gal}$
is the   projective semi-linear group and $\mathcal{S}$ is a subset of $\mathbb{F}_{2^{nr}}$. With the help of this result,
one only needs to count the number of orbits of ${\rm P\Gamma L}$ on $\mathcal{S}$.
The   papers \cite{yueqin},  \cite{magamba},  \cite{musukwa}, \cite{mmr} and \cite{ryan15}  used the Cauchy-Frobenius Theorem to  calculate   the number of orbits of
$\rm{P \Gamma L}$  on  $\mathcal{S}$.
Distinguishing from this approach, we do not use the Cauchy-Frobenius Theorem to obtain our main result;
we introduce an action of
${\rm P\Gamma L}$ on the set of all monic irreducible polynomials of degree $r$ over $\mathbb{F}_q$, say $\mathcal{I}_r$ (see Lemma \ref{action}).
We then show that the number of orbits of  ${\rm P\Gamma L}$ on $\mathcal{S}$ is equal to the number of orbits of ${\rm P\Gamma L}$ on $\mathcal{I}_r$ (see Lemma \ref{orbit}). It allows us to convert the problem of counting the number of orbits of  ${\rm P\Gamma L}$ on $\mathcal{S}$
to that of counting  the number of orbits of ${\rm P\Gamma L}$ on $\mathcal{I}_r$.
We then use a basic  strategy to count the number of orbits of ${\rm P\Gamma L}$ on $\mathcal{I}_r$ (see Lemma \ref{book}): the number of orbits of
${\rm P\Gamma L}$ on $\mathcal{I}_r$ is equal to the number of orbits of ${\rm Gal}$ on ${\rm PGL}\verb|\|\mathcal{I}_r$,
where ${\rm PGL}$ is the projective linear group and ${\rm PGL}\verb|\|\mathcal{I}_r$ is the set of orbits of ${\rm PGL}$ on $\mathcal{I}_r$.
One of the advantages of considering the action of ${\rm P\Gamma L}$ on $\mathcal{I}_r$ rather than that on   $\mathcal{S}$ is that the orbits of
$\langle \sigma^n\rangle$ on ${\rm PGL}\verb|\|\mathcal{I}_r$ are trivial to see, where $\sigma$ is the Frobenius generator of ${\rm Gal}$.

This paper is organized as follows. In Section $2$, we review some definitions and basic results about extended irreducible Goppa codes, some matrix groups and group actions.
In Section  $3$,  we find a formula for the number of orbits of ${\rm P \Gamma L}$ on the set $\mathcal{I}_r$, which naturally gives an upper bound for the number of inequivalent extended irreducible Goppa codes of length $2^n+1$ and degree $r$, where $n$ and $r$ satisfy certain conditions; we also give two small examples to illustrative our main result.
We conclude this paper with remarks and some possible future works in Section $4$.

\section{Preliminaries}
Starting  from this section till the end of this paper, we assume that  $n>3$ is an odd prime number and
$q=2^n$; let $r\geq3$ be a positive integer satisfying $\gcd(r,n)=1$ and $\gcd\big(r,q(q^2-1)\big)=1$.
Let $\mathbb{F}_q$ be the finite field with $q$ elements and let $\mathbb{F}_q^*=\mathbb{F}_q\setminus\{0\}$
be the multiplicative group of the finite field $\mathbb{F}_q$.  Suppose $x$ is an indeterminate over the
finite field $\mathbb{F}_q$ and let $\mathbb{F}_q[x]$ be the polynomial ring in   variable $x$  with coefficients in $\mathbb{F}_q$.
As usual, for a polynomial $f(x)\in \mathbb{F}_q[x]$ (or simply denoted by $f$), $\deg f$ is the degree of $f$;
for a finite set $X$, let $|X|$  denote the number of elements of $X$.

We start by recalling the notion of irreducible binary Goppa codes of length $q$.
For the general definition and more detail information about Goppa codes,   readers may refer to \cite{lx} or \cite{ms}.
\subsection{Extended irreducible Goppa codes}
\begin{Definition}
Let $g(x)$ be a polynomial in $\mathbb{F}_q[x]$ of degree $r$, and let
$L=\mathbb{F}_q=\{\alpha_0,\alpha_1,\cdots,\alpha_{q-1}\}$  such that
$L\cap \{\mbox{zeros of}~g(x)\}=\emptyset$.
The binary Goppa code $\Gamma(L,g)$ of length $q$ and degree $r$ is defined as
$$\Gamma(L,g)=\bigg\{c=(c_0,c_1,\cdots,c_{q-1})\in \mathbb{F}_2^q\,\Big|\,\sum_{i=0}^{q-1}\frac{c_i}{x-\alpha_i}\equiv 0\pmod{g(x)}\bigg\}.$$
The polynomial $g(x)$ is called the Goppa polynomial.
When $g(x)$ is irreducible, $\Gamma(L,g)$ is called an irreducible binary Goppa code of degree $r$.
\end{Definition}
The Goppa code of length $q$ can be extended to a code of length $q+1$ by appending a coordinate in the set $L=\mathbb{F}_q$.
In this paper, we mainly consider extended irreducible Goppa codes.
The definition of  extended irreducible Goppa codes of length $q+1$ and degree $r$ is given below.

\begin{Definition}
For a given monic irreducible polynomial $g(x)$ of degree $r$, let $\Gamma(L,g)$ be an irreducible binary Goppa code
of length $q$ as given above.
The extended Goppa code $\overline{\Gamma(L,g)}$ of length $q+1$ is defined as
$$\overline{\Gamma(L,g)}=\Big\{\big(c_0,c_1,\cdots,c_q\big)\in \mathbb{F}_2^{q+1}\,\Big|\,
\big(c_0,c_1,\cdots,c_{q-1}\big)\in \Gamma(L,g)~\hbox{and}~\sum\limits_{i=0}^qc_i=0\Big\}.$$
\end{Definition}
Chen \cite{Chen} showed that the irreducible binary Goppa code $\Gamma(L,g)$ is completely determined by any root of the Goppa polynomial $g(x)$;
more precisely, if $\alpha$ is a root of $g(x)$ in some extension field over $\mathbb{F}_q$, then
$$
H(\alpha)=\Big(\frac{1}{\alpha-\alpha_0},\frac{1}{\alpha-\alpha_1},\cdots,\frac{1}{\alpha-\alpha_{q-1}}\Big)
$$
can be served as a parity check matrix for $\Gamma(L,g)$.
As such, let $C(\alpha)$ denote the code $\Gamma(L,g)$ and let  $\overline{C(\alpha)}$ denote the code $\overline{\Gamma(L,g)}$.
Therefore, every extended irreducible binary Goppa code of length $q+1$ and degree $r$ can be described as $\overline{C(\alpha)}$
for some $\alpha\in \mathbb{F}_{q^r}$.

\subsection{Equivalent extended irreducible Goppa codes}
The primary purpose of this paper is to give an
upper bound for the number of (inequivalent)
extended irreducible binary   Goppa codes  of length $q+1$ and degree $r$.
This problem can be reduced to that of  counting
the number of orbits of the projective semi-linear group action on some set (see \cite{berger3}, \cite{yueqin} or \cite{ryan15}).
To state this result clearly, we need the notions of group actions and some matrix groups.
In the following we collect the matrix groups that we will
use later and fix the notations.

(1) The general linear group of size $2\times 2$ over $\mathbb{F}_q$
$${\rm GL}={\rm GL}_2(\mathbb{F}_q)=\bigg\{
\begin{pmatrix}
a & b \\
c & d
\end{pmatrix}\bigg|~a,b,c,d\in \mathbb{F}_q, ad-bc\neq 0\bigg\}.$$

(2) The affine general linear group of size $2\times 2$ over $\mathbb{F}_q$
$${\rm AGL}={\rm AGL}_2(\mathbb{F}_q)=\bigg\{
\begin{pmatrix}
a & b \\
0 & 1
\end{pmatrix}\bigg{|}~a\in \mathbb{F}_q^*, b\in \mathbb{F}_q\bigg\}.$$

(3) The projective general linear group of size $2\times 2$ over $\mathbb{F}_q$
$${\rm PGL}={\rm PGL}_2(\mathbb{F}_q)={\rm GL}/\mathcal{Z},$$
where $\mathcal{Z}$ is the center of ${\rm GL}$ consisting of the multiples of the identity matrix by elements of $\mathbb{F}_q^*$.

(4) The projective semi-linear group
$${\rm P\Gamma L}={\rm P\Gamma L}_2(\mathbb{F}_q)={\rm PGL}\rtimes {\rm Gal}=
\Big\{A\sigma^i\,\Big|\,A\in {\rm PGL}, 0\leq i\leq rn-1\Big\},$$
where ${\rm Gal}={\rm Gal}(\mathbb{F}_{q^r}/\mathbb{F}_2)={\rm Gal}(\mathbb{F}_{2^{rn}}/\mathbb{F}_2)=\langle\sigma \rangle$ is
the Galois group of order $rn$ generated by $\sigma$ ($\sigma$ sends each $\alpha\in\mathbb{F}_{q^r}$ to $\alpha^2$).
The operation $``\cdot"$ in ${\rm P\Gamma L}$ is defined as follows:
$$A\sigma^i\cdot B\sigma^j=A\sigma^i(B)\sigma^{i+j},~ 0\leq i,j\leq rn-1,$$
where $\sigma^i(B)=\begin{pmatrix}\sigma^it & \sigma^iu\\ \sigma^iv & \sigma^iw\end{pmatrix}$ for
$B=\begin{pmatrix}t & u\\ v & w\end{pmatrix}\in {\rm PGL}$ ($\sigma^ia$ means $\sigma^ia=a^{2^i}$ for $a\in\mathbb{ F}_{q^r}$).
It is clear that $E_2\sigma^0$ is the identity element of $P\Gamma L$, where $E_2=\left(
                                                                                    \begin{array}{cc}
                                                                                      1 & 0 \\
                                                                                      0 & 1 \\
                                                                                    \end{array}
                                                                                  \right)
$
is the identity matrix.

Now it is the turn of group  actions (for example, see \cite{Rotman}).
Let $\mathcal{S}=\mathcal{S}(r,n)$ denote the set of  elements in $\mathbb{F}_{q^r}$ of degree $r$ over $\mathbb{F}_q$; in other words,
$$\mathcal{S}=\Big\{\alpha\in \mathbb{F}_{q^r}\,\Big|\,\hbox{there exists a monic irreducible polynomial of degree $r$ over $\mathbb{F}_q$ satisfying $f(\alpha)=0$}\Big\}.$$
It is   known that ${\rm PGL}$ and ${\rm P\Gamma L}$ can act on the set $\mathcal{S}$ in the following ways (see \cite{yueqin} or \cite{ryan15}):
\begin{itemize}

\item The action of the projective general linear group on $\mathcal{S}$:
\begin{eqnarray*}
{\rm PGL}\times \mathcal{S} &\rightarrow & \mathcal{S}\\
(A,\alpha) &\mapsto & A\alpha=\frac{a\alpha+b}{c\alpha+d},
\end{eqnarray*}
where $A=\begin{pmatrix}a & b \\ c & d\end{pmatrix}\in {\rm PGL}$.

\item The action of the projective semi-linear group on $\mathcal{S}$:
\begin{eqnarray*}
{\rm P\Gamma L}\times \mathcal{S} &\rightarrow & \mathcal{S}\\
\Big( A\sigma^i,\alpha\Big) &\mapsto & A\sigma^i\alpha=A\big(\sigma^i(\alpha)\big)=\frac{a\sigma^i(\alpha)+b}{c\sigma^i(\alpha)+d}
=\frac{a\alpha^{2^{i}}+b}{c\alpha^{2^{i}}+d}.
\end{eqnarray*}
\end{itemize}
We are ready to recall a sufficient condition
which guarantees two extended irreducible Goppa
codes to be equivalent; thus, in particular, it gives an upper bound for the number of inequivalent codes in
$$
\Big\{\overline{C(\alpha)}\,\Big|\,\alpha\in \mathcal{S}\Big\},
$$
see \cite{berger3}, \cite{yueqin} or \cite{ryan15}.
\begin{lem}\label{important}
Let $\alpha\in \mathcal{S}$ and $\beta\in \mathcal{S}$. If $\alpha,\beta$ lie in the same ${\rm P\Gamma L}$-orbit: $\alpha=A\sigma^i\beta$ for some
$A\sigma^i\in {\rm P\Gamma L}$, then the extended Goppa code $\overline{C(\alpha)}$ is (permutation) equivalent to the extended Goppa code
$\overline{C(\beta)}$.
In particular, the number of inequivalent extended irreducible binary Goppa codes of length $q+1$
and degree $r$ is less than or equal to the number of orbits of ${\rm P\Gamma L}$ on $\mathcal{S}$.
\end{lem}
With the help of Lemma \ref{important},   we only need to count the number of orbits of ${\rm P\Gamma L}$ on $\mathcal{S}$.
\section{An upper bound for the number of extended Goppa codes}
The   papers \cite{yueqin},  \cite{magamba}, \cite{musukwa}, \cite{mmr} and \cite{ryan15} used the Cauchy-Frobenius Theorem to calculate  the number of orbits of
$\rm{P \Gamma L}$  on  $\mathcal{S}$.
Here we introduce another group action: The group $\rm{P \Gamma L}$ can act on the set of all
monic irreducible polynomials of degree $r$ over $\mathbb{F}_q$. Let $\mathcal{I}_r$
be the set of all monic irreducible polynomials of degree $r$ over  $\mathbb{F}_q$.
We will show that the number of orbits of $\rm{P \Gamma L}$ on $\mathcal{S}$ is equal to the number of orbits of
$\rm{P \Gamma L}$ on $\mathcal{I}_r$.

\subsection{The action of ${\rm P\Gamma L}$ on $\mathcal{I}_r$}
Let $A=\begin{pmatrix}a & b \\ c & d\end{pmatrix}\in {\rm PGL}$, $\alpha \in \mathbb{F}_q$
and $f(x)=a_0+a_1x+\cdots+a_rx^r\in \mathbb{F}_q[x]$, where $r\geq 1, a_r\neq 0$. We make the following definitions:
\begin{eqnarray*}
\big(f(x)\big)^*&=&\frac{1}{a_r}f(x),~~
A\alpha=\frac{a\alpha+b}{c\alpha+d},\\
Af&=&(-cx+a)^{r}f\big(A^{-1}x\big)=(-cx+a)^{r}f\Big(\frac{dx-b}{-cx+a}\Big),\\
\sigma^if&=&\sigma^i(f(x))=\sigma^ia_0+\sigma^ia_1x+\cdots+\sigma^ia_rx^r.
\end{eqnarray*}

Now we show that the group $\rm{P\Gamma L}$ acts on the set $\mathcal{I}_r$, as stated below.
\begin{lem}\label{action}
With the notation given above, we have a group action  ${\rm P\Gamma L}$ on the set $\mathcal{I}_r$ defined by
\begin{eqnarray*}
{\rm P\Gamma L}\times \mathcal{I}_r &\rightarrow & \mathcal{I}_r\\
\big(A\sigma^i,f\big) &\mapsto &(A\sigma^i)\big(f\big)=\Big(A(\sigma^if)\Big)^*.
\end{eqnarray*}
\end{lem}

\begin{proof}
Verification of the group action conditions is routine.
We first introduce the strategy: Let $\mathcal{X}_r$ denote the set of all irreducible polynomials of degree $r$ over $\mathbb{F}_q$.
There is an equivalence relation defined on $\mathcal{X}_r$: $f\sim g$ if and only if there exists $\lambda\in \mathbb{F}_q^*$
such that $f=\lambda g$. We claim that
\begin{eqnarray*}
{\rm P\Gamma L}\times \mathcal{X}_r &\rightarrow & \mathcal{X}_r\\
\big(A\sigma^i,f\big) &\mapsto &(A\sigma^i)\big(f\big)=A(\sigma^if)
\end{eqnarray*}
defines an action of ${\rm P\Gamma L}$ on $\mathcal{X}_r$.
Once this claim is established, one can easily show that if $f\sim g$ then $A\sigma^if\sim A\sigma^ig$; this implies that
${\rm P\Gamma L}$ acts on the equivalence classes $\mathcal{X}_r/\sim$. Since every equivalence class contains a unique monic polynomial, it is trivial to see that ${\rm P\Gamma L}$ acts on $\mathcal{I}_r$ in the way stated in the lemma. Therefore, it is enough to verify that
${\rm P\Gamma L}$ acts on $\mathcal{X}_r$.

We will give a detail proof  by carrying out the following steps, although it is somewhat tedious.

Step $1.$ ${\rm deg}\big(A(\sigma^if)\big)={\rm deg}\big( f\big)$ for any $f\in \mathcal{X}_r$.

Suppose that $f(x)=f_0+f_1x+\cdots+f_{r-1}x^{r-1}+f_rx^r$, where $f_r\neq0$. Then
\begin{eqnarray*}
A\big((\sigma^if)(x)\big) & = &  (\sigma^if)\Big(\frac{dx-b}{-cx+a}\Big)\cdot (-cx+a)^r\\
&=&\sigma^if_{r}(dx-b)^r+\sigma^if_{n-1}(dx-b)^{r-1}(-cx+a)+\cdots+\sigma^if_0(-cx+a)^r\\
&=& x^r\big[\sigma^if_{r}d^r+\sigma^if_{r-1}(-c)d^{r-1}+\cdots+\sigma^if_0(-c)^r\big]+\cdots.
\end{eqnarray*}

If $c\neq 0$, then
\begin{eqnarray*}
&&\sigma^if_{r}d^r+\sigma^if_{r-1}(-c)d^{r-1}+\cdots+\sigma^if_0(-c)^r\\
&=& (-c)^r\Big[\sigma^if_0+\sigma^if_1\Big(-\frac{d}{c}\Big)+\cdots+
\sigma^if_{r-1}\Big(-\frac{d}{c}\Big)^{r-1}+\sigma^if_{r}\Big(-\frac{d}{c}\Big)^{r}\Big]\\
&=& (-c)^r(\sigma^if)\Big(-\frac{d}{c}\Big).
\end{eqnarray*}
Noting that $c,d\in \mathbb{F}_q$ and $\sigma^if$ is irreducible over $\mathbb{F}_q$, we obtain
$(-c)^r(\sigma^if)\big(-\frac{d}{c}\big)\neq 0$.

If $c=0$, then $A\in {\rm PGL}$, which yields $d\neq 0$, and we have
$\sigma^if_rd^r\neq0$, i.e., the  leading coefficient of $x^r$ in $(A\sigma^i)f$ is nonzero.

In conclusion, we have
${\rm deg}\big(A(\sigma^if)\big)={\rm deg}\big( f\big)$ for any $f(x)\in \mathcal{X}_r$.

Step $2.$ $A(\sigma^if)\in \mathcal{X}_r$ for any $f\in \mathcal{X}_r$.

It is enough to show that $A(\sigma^if)$ is irreducible over $\mathbb{F}_q$.
To this end, we use the following result (see \cite[Proposition 4.13]{Rotman}):
Let $k$ be a field and let $p(x)\in k[x]$ have no repeated roots.
If $E/k$ is a splitting field of $p(x)$, then $p(x)$ is irreducible if and only if the Galois group of $E$ over $k$, denoted by ${\rm Gal}(E/k)$, acts transitively on the roots of $p(x)$.

Suppose that $\alpha,\alpha^q,\cdots,\alpha^{q^{r-1}}$ are all the distinct roots of $f(x)$.
Then $A(\sigma^i\alpha),A(\sigma^i\alpha^q),\cdots,A(\sigma^i\alpha^{q^{r-1}})$ are   all the roots of $A(\sigma^if)$.
Since for $j=0,1,2,\cdots,r-1$,
$$A(\sigma^i\alpha^{q^{j}})=\frac{a\sigma^i\alpha^{q^{j}}+b}{c\sigma^i\alpha^{q^{j}}+d}
=\Big(\frac{a\sigma^i\alpha+b}{c\sigma^i\alpha+d}\Big)^{q^{j}}=\big(A(\sigma^i\alpha)\big)^{q^j},$$
$A(\sigma^i\alpha),A(\sigma^i\alpha^q),\cdots,A(\sigma^i\alpha^{q^{r-1}})$ are distinct; that is to say
$A(\sigma^if)$ has no repeated roots.
On the other hand, $A(\sigma^if)\in \mathbb{F}_q[x]$ and the splitting field of $A(\sigma^if)$
is $\mathbb{F}_q\big(A(\sigma^i\alpha)\big)$.

Write $\beta=A(\sigma^i\alpha)$. Then $\alpha=\sigma^{rn-i}(A^{-1}\beta)$. So
$$\mathbb{F}_q\big(A(\sigma^i\alpha)\big)=\mathbb{F}_q(\beta)=\mathbb{F}_q(\alpha)=\mathbb{F}_{q^r}.$$
Hence $\mathbb{F}_{q^r}$ is the splitting field of $A(\sigma^if)$.

Clearly, ${\rm Gal}(\mathbb{F}_{q^r}/\mathbb{F}_q)=\langle \sigma^n\rangle$. Let $\tau=\sigma^n$ and then all the distinct roots of $A(\sigma^if)$ are
$$\beta, \tau(\beta), \tau^2(\beta),\cdots,\tau^{r-1}(\beta).$$
Thus ${\rm Gal}(\mathbb{F}_{q^r}/\mathbb{F}_q)$ acts transitively on the roots of $A(\sigma^if)$.
Hence according to \cite[Proposition 4.13]{Rotman} $A(\sigma^if)$ is irreducible over $\mathbb{F}_q$.

Step $3.$ Clearly, $(E_2\sigma^0)f=E_2(\sigma^0f)=f$, where $E_2\sigma^0$ is the identity element of the group $P\Gamma L$.

Step $4.$ We are left to check that
$$(A\sigma^i)\big[(B\sigma^j)f\big]=\big[(A\sigma^i)\cdot(B\sigma^j)\big]f,$$
where $A=\begin{pmatrix}a & b \\ c & d\end{pmatrix}\in {\rm PGL}$, $B=\begin{pmatrix}t & u\\ v & w\end{pmatrix}\in {\rm PGL}$.
On the one hand,
\begin{eqnarray*}
(A\sigma^i)\big[(B\sigma^j)f\big]
&=&(A\sigma^i)\Big[(-vx+t)^r(\sigma^jf)\Big(\frac{wx-u}{-vx+t}\Big)\Big]\\
&=& (-cx+a)^r\Big(-\sigma^iv\frac{dx-b}{-cx+a}+\sigma^it\Big)^r(\sigma^{i+j}f)
\bigg(\frac{\sigma^iw\frac{dx-b}{-cx+a}-\sigma^iu}{-\sigma^iv\frac{dx-b}{-cx+a}+\sigma^it}\bigg)\\
&=&\big[-\sigma^iv(dx-b)+\sigma^it(-cx+a)\big]^r(\sigma^{i+j}f)\bigg(\frac{\sigma^iw(dx-b)-\sigma^iu(-cx+a)}{-\sigma^iv(dx-b)+\sigma^it(-cx+a)}\bigg).
\end{eqnarray*}

On the other hand, since
\begin{eqnarray*}
(A\sigma^i)\cdot(B\sigma^j)
&=&A\sigma^i(B)\sigma^{i+j}\\
&=&\begin{pmatrix}a & b \\ c & d\end{pmatrix}\begin{pmatrix}\sigma^it & \sigma^iu \\ \sigma^iv & \sigma^iw\end{pmatrix}\sigma^{i+j}\\
&=&\begin{pmatrix}a\sigma^it+b\sigma^iv & a\sigma^iu+b\sigma^iw \\ c\sigma^it+d\sigma^iv & c\sigma^iu+d\sigma^iw\end{pmatrix}\sigma^{i+j}
\end{eqnarray*}
and
$$\big(A\sigma^i(B)\big)^{-1}=\begin{pmatrix}c\sigma^iu+d\sigma^iw & -a\sigma^iu-b\sigma^iw\\
-c\sigma^it-d\sigma^iv & a\sigma^it+b\sigma^iv\end{pmatrix},$$
we have
\begin{eqnarray*}
\big[(A\sigma^i)\cdot(B\sigma^j)\big]f
&=&\bigg(\begin{pmatrix}a\sigma^it+b\sigma^iv & a\sigma^iu+b\sigma^iw \\ c\sigma^it+d\sigma^iv & c\sigma^iu+d\sigma^iw\end{pmatrix} \sigma^{i+j}\bigg)f\\
&=& \big(-(c\sigma^it+d\sigma^iv)x+a\sigma^it+b\sigma^iv\big)^r\cdot(\sigma^{i+j}f)
\bigg(\frac{(c\sigma^iu+d\sigma^iw)x-(a\sigma^iu+b\sigma^iw)}{-(c\sigma^it+d\sigma^iv)x+a\sigma^it+b\sigma^iv}\bigg).
\end{eqnarray*}
Then we obtain  $(A\sigma^i)\big[(B\sigma^j)f\big]=\big[(A\sigma^i)\cdot(B\sigma^j)\big]f$,
as wanted.
\end{proof}
\begin{Remark}{\rm
Many authors have studied the action of ${\rm PGL}$ on $\mathcal{I}_r$, focusing on the characterization
and number of $A$-invariants where $A\in{\rm PGL}$ (for example, see \cite{Gare}, \cite{Reis18}, \cite{Reis182}, \cite{Reis20}, \cite{ST}).
The paper \cite{MOR} considered an action of ${\rm P\Gamma L}$ on $\mathcal{I}_r$, and our definition of ${\rm P\Gamma L}$ on $\mathcal{I}_r$ is different from  that of \cite{MOR}}.
\end{Remark}

\subsection{The orbits of ${\rm P\Gamma L}$ on $\mathcal{I}_r$}
In this subsection we analyze the orbits of ${\rm P\Gamma L}$ on $\mathcal{I}_r$. We first introduce some notations.
For a general group $G$ acting on a set $X$, let $G(x)$ denote the orbit containing $x\in X$, namely $G(x)=\{gx\,|\,g\in G\}$; let
${\rm Stab}_G(x)$ be the stabilizer of the point $x\in X$ in $G$, namely  $\rm{{Stab}}$$_G(x)=\{g\in G\,|\,gx=x\}$.
For example, ${\rm PGL}(\alpha)=\big\{  A\alpha\,\big|\,  A\in {\rm PGL}\big\}$   denotes the orbit  of $\alpha\in \mathcal{S}$ under the action of
${\rm PGL}$ on $\mathcal{S}$, and
${\rm P\Gamma L}(f)=\big\{(A\sigma^i)\big(f\big)\,\big|\,A\in {\rm PGL}, 0\leq i\leq rn-1\big\}$
denotes the orbit  of $f\in \mathcal{I}_r$ under the action of ${\rm P\Gamma L}$ on $\mathcal{I}_r$.

The next result reveals that the problem of counting the number of orbits of  ${\rm P\Gamma L}$ on $\mathcal{S}$ can be completely converted
to the problem of counting  the number of orbits of ${\rm P\Gamma L}$ on $\mathcal{I}_r$.
\begin{lem}\label{orbit}
The number of orbits of  ${\rm P\Gamma L}$ on $\mathcal{S}$ is equal to the number of orbits of ${\rm P\Gamma L}$ on $\mathcal{I}_r$.
\end{lem}
\begin{proof}

Let ${\rm P\Gamma L}\verb|\|\mathcal{S}$ be the set of  orbits of ${\rm P\Gamma L}$ on $\mathcal{S}$ and let
${\rm P\Gamma L}\verb|\|\mathcal{I}_r$ be the set of  orbits of ${\rm P\Gamma L}$ on $\mathcal{I}_r$.
To prove $|{\rm P\Gamma L}\verb|\|\mathcal{S}|=|{\rm P\Gamma L}\verb|\|\mathcal{I}_r|$, it suffices to show that there is a bijection between ${\rm P\Gamma L}\verb|\|\mathcal{S}$ and ${\rm P\Gamma L}\verb|\|\mathcal{I}_r$.
Define a map $\varphi$ as follows:
\begin{eqnarray*}
\varphi: {\rm P\Gamma L}\verb|\|\mathcal{I}_r &\rightarrow & {\rm P\Gamma L}\verb|\|\mathcal{S}\\
{\rm P\Gamma L}(f) &\mapsto & \varphi\big({\rm P\Gamma L}(f)\big),
\end{eqnarray*}
where
$$\varphi\big({\rm P\Gamma L}(f)\big)=\big\{\alpha\,\big|\,\mbox{there exists a polynomial}~g(x)\in{\rm P\Gamma L}(f) ~\mbox{such that}~g(\alpha)=0\big\}.$$

We will show that $\varphi$ is a bijection by carrying out the following steps.

(1) $\varphi\big({\rm P\Gamma L}(f)\big)\in {\rm P\Gamma L}\verb|\|\mathcal{S}$.
Let $f(\alpha)=0$, i.e., $\alpha$ is a root of $f(x)$. Then
$(A\sigma^i f)\big(A(\sigma^i(\alpha))\big)=0.$
i.e., $A(\sigma^i(\alpha))$ is a root of $(A\sigma^i)\big(f(x)\big)$.
Assume that $\alpha, \alpha^q, \cdots,\alpha^{q^{r-1}}$ are all the distinct roots of $f(x)$.
Based on the above fact we have that
$$A(\sigma^i(\alpha)), A(\sigma^i(\alpha^q)),\cdots,A(\sigma^i(\alpha^{q^{r-1}}))$$
are all the distinct roots of the polynomial $(A\sigma^i)\big(f(x)\big)$.
Noting that $$A(\sigma^i(\alpha^{q^j}))=\big(A(\sigma^i\alpha)\big)^{q^j}=A(\sigma^{i+nj}\alpha), ~~j=0,1,\cdots,r-1,$$
we obtain that
\begin{eqnarray*}
\varphi\big({\rm P\Gamma L}(f)\big)&=&\big\{\big(A(\sigma^i\alpha)\big)^{q^j}\,\big|\,A\in {\rm GL}_,0\leq i\leq rn-1, 0\leq j\leq r-1\big\}\\
&=&\big\{A(\sigma^{i+nj}\alpha)\,\big|\,A\in {\rm GL},0\leq i\leq rn-1, 0\leq j\leq r-1\big\}\\
&=& \big\{A(\sigma^{i}\alpha)\,\big|\,A\in {\rm GL},0\leq i\leq rn-1\big\}\\
&=& {\rm P\Gamma L}(\alpha),
\end{eqnarray*}
which shows that $\varphi\big({\rm P\Gamma L}(f)\big)\in {\rm P\Gamma L}\verb|\|\mathcal{S}$.

(2) $\varphi$ is injective.
According to the definition of $\varphi$ it is obvious that $\varphi$ is injective.

(3) $\varphi$ is surjective.
Given an orbit $\Lambda=\rm{P\Gamma L}(\alpha)$ in ${\rm P\Gamma L}\verb|\|\mathcal{S}$, i.e., $\alpha\in \mathcal{S}$ and
$$\Lambda=\big\{  A(\sigma^{i}\alpha)\,\big|\,  A\in {\rm PGL},0\leq i\leq rn-1\big\}.$$
Suppose that $f(x)$ is an irreducible polynomial of degree $r$ over $\mathbb{F}_q$ with $f(\alpha)=0$, then
$\varphi\big({\rm P\Gamma L}(f)\big)=\Lambda.$
Thus $\varphi$ is surjective.
\end{proof}
By Lemma \ref{orbit},
our ultimate aim is to find the number of   orbits of  ${\rm P\Gamma L}={\rm PGL}\rtimes {\rm Gal}$ on the set $\mathcal{I}_r$.
For this purpose, we will repeatedly use the following fact to achieve this goal (for example, see \cite[Pages 35-36]{Kerber}):
\begin{lem}\label{book}
Let $G$ be a finite group acting on a finite set $X$ and let $N$ be a normal subgroup of $G$. It is clear that $N$ naturally acts on $X$. Suppose the $N$-orbits are denoted by $N\verb|\|X=\{N(x)\,|\,x\in X\}$. Then the factor group $G/N$ acts on $N\verb|\|X$ and the number of orbits of $G$ on $X$
is equal to the number of orbits of $G/N$ on $N\verb|\|X$.
\end{lem}

As $\rm{PGL}$ is a normal subgroup of $\rm{P\Gamma L}$, by virtue of  Lemma \ref{book}
we  first count the number of orbits of $\rm{PGL}$ on the set $\mathcal{I}_r$.
The next result shows that if $\gcd\big(r, q(q^2-1)\big)=1$, then the size of  each orbit of $\rm{PGL}$ on   $\mathcal{I}_r$
is equal to $q(q^2-1)$; in other words, ${\rm Stab}_{{\rm PGL}}(f)=\{E_2\}$ for any $f\in \mathcal{I}_r$, where $E_2=\left(
                                                                                                                       \begin{array}{cc}
                                                                                                                         1 & 0 \\
                                                                                                                         0 & 1 \\
                                                                                                                       \end{array}
                                                                                                                     \right)
$
is the identity matrix.
\begin{lem}\label{orbitsize1}
Let $r$ be a positive integer satisfying $r\geq 3$ and $\gcd\big(r, q(q^2-1)\big)=1$. Let $f(x)\in \mathcal{I}_r$ and
$${\rm Stab}_{{\rm PGL}}(f)=\big\{A\in {\rm PGL}\,\big|\,Af=f\big\}.$$
Then ${\rm Stab}_{{\rm PGL}}(f)=\{E_2\}$.
\end{lem}
\begin{proof}
Fix an irreducible polynomial $f(x)$ of degree $r$ over $\mathbb{F}_q$, i.e., $f(x)\in \mathcal{I}_r$.
Suppose  $A\in {\rm Stab}_{{\rm PGL}}(f)$ and the order of $A\in{\rm PGL}$ is equal to $\ell$,
where $A=\begin{pmatrix}a & b \\ c & d \end{pmatrix}$.
In the following we want to  prove that $A$ must be the identity matrix, i.e., $A=E_2$.
Let $\alpha$ be a root of $f(x)$. Then $  A\alpha$ is a root of $Af$.
Thus $A\alpha$  is also a root of $f(x)$.
Hence there exists a positive $s$ satisfying   $0\leq s\leq r-1$ such that $A\alpha=\alpha^{q^s}$,
thus $A^\ell\alpha=\alpha^{q^{s\ell}}$.
Since $A^\ell=E_2$, we have $A^\ell\alpha=E_2\alpha=\alpha$, therefore
$\alpha^{q^{s\ell}}=\alpha.$
This shows that $\mathbb{F}_{q^r}=\mathbb{F}_q(\alpha)\subseteq \mathbb{F}_{q^{s\ell}}$,
which gives that $r$ is a divisor of $s\ell$.

Clearly, $\ell$ is a divisor of $q(q^2-1)$. Since $\gcd\big(r, q(q^2-1)\big)=1$, we obtain that $\gcd(r,\ell)=1$.
Hence $r$ is a divisor of $s$. We know that $\alpha^{q^r}=\alpha$, and we get  $\alpha^{q^s}=\alpha$; it follows
from $A\alpha=\alpha^{q^s}$ that $A\alpha=\alpha$.
Substitute  $A$ to $\begin{pmatrix}a & b \\ c & d \end{pmatrix}$, and it leads to
$$\frac{a\alpha+b}{c\alpha+d}=\alpha,$$
that is to say
$c\alpha^2+(d+a)\alpha+b=0.$
By our assumption  $r\geq 3$, the above equality tells us that $c=b=0$ and $a=d$, i.e.,  $A=E_2$.
This concludes the proof.
\end{proof}
Lemma \ref{orbitsize1} shows that the size of every orbit ${\rm PGL}(f)$ of    ${\rm PGL}$  on $\mathcal{I}_r$ is equal to
$$|{\rm PGL}(f)|=\big[{\rm PGL}: {\rm Stab}_{{\rm PGL}}(f)\big]=\big|{\rm PGL}\big|=q(q^2-1).$$
Let ${\rm PGL}\verb|\|\mathcal{I}_r$
be the set of all  orbits of ${\rm PGL}$ on $\mathcal{I}_r$.
It follows from Lemma \ref{orbitsize1} and the enumerative formula for the size of $\mathcal{I}_r$ (see \cite[Theorem 3.25]{lidl}) that
\begin{equation}\label{ir}
\big|{\rm PGL}\verb|\|\mathcal{I}_r\big|=\frac{|\mathcal{I}_r|}{q(q^2-1)}=\frac{\sum\limits_{d\,|\,r}\mu(d)q^{r/d}}{rq(q^2-1)},
\end{equation}
where $\mu$ is the M\"{o}bius function.

\subsection{The orbits of ${\rm Gal}$ on ${\rm PGL}\backslash\mathcal{I}_r$}
In order to get the number of orbits of ${\rm P\Gamma L}$ on $\mathcal{I}_r$, by Lemmas \ref{book} and \ref{orbitsize1}, we have to count the number of orbits of   ${\rm Gal}$ on ${\rm PGL}\verb|\|\mathcal{I}_r$.
Recall that   the Galois group ${\rm Gal}={\rm Gal}(\mathbb{F}_{q^r}/\mathbb{F}_2)={\rm Gal}(\mathbb{F}_{2^{rn}}/\mathbb{F}_2)=\langle\sigma \rangle$ is the cyclic group of order $rn$ generated by $\sigma.$
The action of  ${\rm Gal}$   on ${\rm PGL}\verb|\|\mathcal{I}_r$ is given by
\begin{eqnarray*}
{\rm Gal}\times {\rm PGL}\verb|\|\mathcal{I}_r &\rightarrow & {\rm PGL}\verb|\|\mathcal{I}_r\\
\big(\sigma^i, {\rm PGL}(f)\big)&\mapsto & \sigma^i\big({\rm PGL}(f)\big)={\rm PGL}(\sigma^if).
\end{eqnarray*}
Recall also that $n\geq 3$ is a prime number, $q=2^n$ and  $r$ is a positive integer satisfying $\gcd(r,n)=1$ and $\gcd\big(r, q(q^2-1)\big)=1$.
Since $\gcd(r,n)=1$, ${\rm Gal}=\langle\sigma\rangle$ has the following decomposition into direct products:
$${\rm Gal}=\langle\sigma^r\rangle\times \langle\sigma^n\rangle.$$
In order to count the number of orbits of   ${\rm Gal}$ on ${\rm PGL}\verb|\|\mathcal{I}_r$, using Lemma \ref{book} again we first consider the
action of $\langle\sigma^n\rangle$ on  ${\rm PGL}\verb|\|\mathcal{I}_r$ ($\langle\sigma^n\rangle$ is certainly a normal subgroup of ${\rm Gal}$).
Note that the action of $\langle\sigma^n\rangle$ on ${\rm PGL}\verb|\|\mathcal{I}_r$ is given by
\begin{eqnarray*}
\langle \sigma^n\rangle\times {\rm PGL}\verb|\|\mathcal{I}_r &\rightarrow & {\rm PGL}\verb|\|\mathcal{I}_r\\
\big(\sigma^{ni}, {\rm PGL}(f)\big)&\mapsto & \sigma^{ni}\big({\rm PGL}(f)\big)={\rm PGL}(\sigma^{ni}f).
\end{eqnarray*}
Observe that
$\sigma^na=a^{2^n}=a^q=a \hbox{ for any $a\in \mathbb{F}_q$,}$
which gives
$${\rm PGL}(\sigma^{ni}f)={\rm PGL}(f)~~\hbox{for any $f\in \mathcal{I}_r$}.$$
That is to say that $\langle\sigma^n\rangle$ fixes each ${\rm PGL}(f)$ in ${\rm PGL}\verb|\|\mathcal{I}_r$; in other words,
the set of orbits of $\langle\sigma^n\rangle$ on ${\rm PGL}\verb|\|\mathcal{I}_r$  remains ${\rm PGL}\verb|\|\mathcal{I}_r$.
By Lemma \ref{book}, the number of orbits of ${\rm Gal}$ on ${\rm PGL}\verb|\|\mathcal{I}_r$ is equal to the
number of orbits of $\langle \sigma^r\rangle$ on ${\rm PGL}\verb|\|\mathcal{I}_r$.
Since $\langle\sigma^r\rangle$ is of prime order $n$, the size of every orbit of $\langle \sigma^r\rangle$ on ${\rm PGL}\verb|\|\mathcal{I}_r$ is equal to $1$ or $n$.
Thus it is enough to determine the number of  orbits of $\langle \sigma^r\rangle$ on ${\rm PGL}\verb|\|\mathcal{I}_r$ with size $1$.


\begin{lem}\label{equivalent}
Let notation be the same as before. Assume that $r\geq3$ is a positive integer satisfying
$\gcd\big(r,q(q^2-1)\big)=1$. Let $f\in \mathcal{I}_r$ and let $\alpha$ be a root of $f(x)$. Then
${\rm PGL}(\sigma^r f)={\rm PGL}(f)$
if and only if
${\rm PGL}(\sigma^r\alpha)={\rm PGL}(\alpha).$
\end{lem}
\begin{proof}
Suppose ${\rm PGL}(\sigma^r f)={\rm PGL}(f)$. Then there exists $A\in {\rm PGL}$ such that
$A(\sigma^rf)=f$. This implies that  $A(\sigma^r\alpha)$ is also a root of $f(x)$.
Thus there is an integer $s\geq 0$ satisfying
\begin{equation}\label{equation1}
A(\sigma^r\alpha)=\alpha^{q^s}.
\end{equation}
Simple algebraic calculations in ${\rm P\Gamma L}$ show that
$$(A\sigma^r)^n=A\sigma^r(A)\sigma^{2r}(A)\cdots\sigma^{r(n-1)}(A)\sigma^{rn}
=A\sigma^r(A)\sigma^{2r}(A)\cdots\sigma^{r(n-1)}(A)\sigma^0.$$
Let $\ell$ be the order of $A\sigma^r(A)\sigma^{2r}(A)\cdots\sigma^{r(n-1)}(A)$
in the group ${\rm PGL}$.
Then
$(A\sigma^r)^{\ell n}=E_2\sigma^0.$
Hence, on the one hand,
$(A\sigma^r)^{\ell n}\alpha=(E_2\sigma^0)\alpha=\alpha;$
on the other hand, from (\ref{equation1}) we have that
$(A\sigma^r)^{\ell n}\alpha=\alpha^{q^{\ell ns}}.$
Therefore
$\alpha^{q^{\ell ns}}=\alpha,$
yielding that $r$ is a divisor of  $\ell ns$.
Noting that $\ell$ is a divisor of $q(q^2-1)$, we get that $\gcd(r,\ell)=1$. Thus $\gcd(r, \ell n)=1$.
It follows that $r$ is a divisor of $s$.
Therefore (\ref{equation1}) becomes
$A(\sigma^r\alpha)=\alpha^{q^s}=\alpha,$
which implies   ${\rm PGL}(\sigma^r\alpha)={\rm PGL}(\alpha)$.

Conversely, suppose that ${\rm PGL}(\sigma^r\alpha)={\rm PGL}(\alpha)$.
Then there is a matrix $A\in {\rm PGL}$ such that $A(\sigma^r\alpha)=\alpha$.
Note that $A(\sigma^r\alpha)$ is a root of $A(\sigma^rf)$, and then we obtain
${\rm PGL}(\sigma^r f)={\rm PGL}(f).$
We are done.
\end{proof}
To count the number of
${\rm PGL}(f)\in {\rm PGL}\verb|\|\mathcal{I}_r$ that are fixed by $\langle \sigma^r\rangle$, by virtue of \cite{ryan15} we need to use the
affine general linear group ${\rm AGL}$. The affine general linear group ${\rm AGL}$ can be viewed naturally as a subgroup of ${\rm PGL}$.
Hence, the group ${\rm AGL}$ acts on the set $\mathcal{S}$ naturally.
Let
$$
{\rm AGL}\verb|\|\mathcal{S}=\big\{{\rm AGL}(\alpha)\,\big|\,\alpha\in \mathcal{S}\big\}
$$
be the set of all orbits of ${\rm AGL}$ on $\mathcal{S}$. Then the cyclic group $\langle \sigma^r\rangle$ acts on
${\rm AGL}\verb|\|\mathcal{S}$ in the following way:

\begin{equation}\label{AGL-action}
\langle \sigma^r\rangle\times {\rm AGL}\verb|\|\mathcal{S}  \rightarrow  {\rm AGL}\verb|\|\mathcal{S},~~~~
\\
\big(\sigma^{ri}, {\rm AGL}(\alpha)\big)\mapsto  \sigma^{ri}\big({\rm AGL}(\alpha)\big)={\rm AGL}(\sigma^{ri}\alpha).
\end{equation}
It is not hard to verify that this is indeed a group action.
We now turn to consider the orbit ${\rm PGL}(\alpha)$ where $\alpha\in \mathcal{S}$.
There is an action of ${\rm AGL}$ on ${\rm PGL}(\alpha)$:
\begin{eqnarray*}
{\rm AGL}\times {\rm PGL}(\alpha) &\rightarrow & {\rm PGL}(\alpha)\\
(C,  A\alpha)&\mapsto & CA\alpha.
\end{eqnarray*}
Therefore, ${\rm PGL}(\alpha)$ is the disjoint union of ${\rm AGL}$-orbits. Indeed, one can easily check that there are exactly
$q+1$ right cosets of ${\rm AGL}$ in ${\rm PGL}$ and
$$
t_0=E_2=\left(
                         \begin{array}{cc}
                           1 &0 \\
                           0 & 1 \\
                         \end{array}
                       \right),~~ t_1=\left(
                         \begin{array}{cc}
                           0 & 1 \\
                           1 & 0 \\
                         \end{array}
                       \right)
\hbox{~~and~~}
t_\gamma=\left(
                         \begin{array}{cc}
                           0 & 1 \\
                           1 & \gamma \\
                         \end{array}
                      \right)    \hbox{~~for any $\gamma\in \mathbb{F}_q^*$}
$$
consists of a right coset representative of ${\rm AGL}$ in ${\rm PGL}$.
The coset decomposition
$${\rm PGL}={\rm AGL}t_0\bigcup{\rm AGL}t_1\bigcup_{\gamma\in \mathbb{F}_q^*}{\rm AGL}t_\gamma$$
gives rise to the orbit decomposition of ${\rm PGL}(\alpha)$ into ${\rm AGL}$-orbits
$${\rm PGL}(\alpha)={\rm AGL}(t_0\alpha)\bigcup{\rm AGL}(t_1\alpha)\bigcup_{\gamma\in \mathbb{F}_q^*}{\rm AGL}(t_\gamma\alpha).$$
We have arrived at the following result, which has been appeared previously in \cite{yueqin} and \cite{ryan15}.
\begin{lem}\label{partition}
Let $\alpha\in \mathcal{S}$. Then
$${\rm PGL}(\alpha)=\bigcup_{\gamma\in \mathbb{F}_q}{\rm AGL}\Big(\frac{1}{\alpha+\gamma}\Big)\bigcup{\rm AGL}(\alpha),$$
is a partition of ${\rm PGL}(\alpha)$ into ${\rm AGL}$-orbits.
\end{lem}
Lemma \ref{partition} implies that
\begin{equation*}
\begin{split}
{\rm PGL}\big(\sigma^r(\alpha)\big)&=\bigcup_{\gamma\in
\mathbb{F}_q}{\rm AGL}\Big(\frac{1}{\sigma^r(\alpha)+\gamma}\Big)\bigcup{\rm AGL}\big(\sigma^r(\alpha)\big)\\
&=\bigcup_{\gamma\in
\mathbb{F}_q}{\rm AGL}\Big(\frac{1}{\sigma^r(\alpha)+\sigma^r(\gamma)}\Big)\bigcup{\rm AGL}\big(\sigma^r(\alpha)\big)\\
&=\bigcup_{\gamma\in
\mathbb{F}_q}{\rm AGL}\Big(\sigma^r\Big(\frac{1}{\alpha+\gamma}\Big)\Big)\bigcup{\rm AGL}\big(\sigma^r(\alpha)\big).\\
\end{split}
\end{equation*}
Suppose now that ${\rm PGL}(\alpha)$ is fixed by the cyclic group $\langle\sigma^r\rangle$,
i.e., ${\rm PGL}(\sigma^r\alpha)={\rm PGL}(\alpha)$.
In this case, the cyclic group $\langle \sigma^r\rangle$ acts on the set of ${\rm AGL}$-orbits
$${\rm AGL}\verb|\|{\rm PGL}(\alpha)=\Big\{{\rm AGL}(\alpha), {\rm AGL}\Big(\frac{1}{\alpha+\gamma}\Big)\,\Big|\,\gamma\in \mathbb{F}_q\Big\}$$
in the way given in (\ref{AGL-action}).
The next result  has been appeared in \cite{yueqin} and \cite{ryan15}.
\begin{lem}\label{fixpoint}
Let $n>3$ be a prime number. If ${\rm PGL}(\sigma^r\alpha)={\rm PGL}(\alpha)$,
then there exists a fixed point of $\langle\sigma^r\rangle$  on ${\rm AGL}\verb|\|{\rm PGL}(\alpha)$. In other words,
either ${\rm AGL}(\sigma^r\alpha)={\rm AGL}(\alpha)$ or
${\rm AGL}\big(\sigma^r(\frac{1}{\alpha+\gamma})\big)={\rm AGL}\big(\frac{1}{\alpha+\gamma}\big)$
for some $\gamma\in \mathbb{F}_q$.
\end{lem}
\begin{proof}
Note that the following properties hold:
(1) $|{\rm AGL}\verb|\|{\rm PGL}(\alpha)|=q+1$; (2) $n$ does not divide $q+1$;
(3) the size of each orbit of $\langle \sigma^r\rangle$ on ${\rm AGL}\verb|\|{\rm PGL}(\alpha)$ is either $1$ or $n$.
We get the required result.
\end{proof}

The following result is crucial to our enumeration.
\begin{lem}\label{nexttolastlem}
Let notation be the same as before. Let $f\in \mathcal{I}_r$.
Then
${\rm PGL}(\sigma^rf)={\rm PGL}(f)$
if and only if there is a polynomial $g(x)\in {\rm PGL}(f)$ such that
$g(x)$ divides $x^{2^r}+x.$
\end{lem}
\begin{proof}
To complete the proof,
we mainly modify the arguments in \cite[Lemma 3.1 and Theorem 3.2]{ryan15}. Let $\alpha$ be a root of $f(x)$.
Lemma \ref{equivalent} says that ${\rm PGL}(\sigma^r f)={\rm PGL}(f)$
if and only if
${\rm PGL}(\sigma^r\alpha)={\rm PGL}(\alpha).$

Suppose that ${\rm PGL}(\sigma^r\alpha)={\rm PGL}(\alpha)$, then according to Lemma \ref{fixpoint}, we need to consider two cases:

(1) ${\rm AGL}(\sigma^r\alpha)={\rm AGL}(\alpha)$.
In this case there are $\theta\neq0, \tau\in \mathbb{F}_q$ such that $\alpha^{2^r}=\theta\alpha+\tau.$
Let $\rho$ be a primitive element of $\mathbb{F}_q$, i.e., $\mathbb{F}_q^*=\langle\rho\rangle$ and $\rho$ is of order $q-1=2^n-1$.
Since $(2^r-1, 2^n-1)=1$, $\mathbb{F}_q^*=\langle\rho^{2^r-1}\rangle$.
Assume that $\theta=\rho^{(2^r-1)\kappa}$, where $\kappa$ is a positive integer.
Then there exists an element $\mu=\rho^{-\kappa}\in \mathbb{F}_q$ satisfying $\mu^{2^r-1}\theta=1.$
It follows that $\mu\alpha+\upsilon\in {\rm AGL}(\alpha)$ for each $v\in \mathbb{F}_q$ and
\begin{eqnarray*}
(\mu\alpha+\upsilon)^{2^r}& = &  \mu^{2^r}\alpha^{2^r}+\upsilon^{2^r}\\
& = &  \mu^{2^r}(\theta\alpha+\tau)+v^{2^r}\\
&=& \mu^{2^r-1}\theta(\mu\alpha+\upsilon)+\mu^{2^r-1}\theta\upsilon+\mu^{2^r}\tau+\upsilon^{2^r}\\
&=& (\mu\alpha+\upsilon)+(\upsilon+\mu^{2^r}\tau+\upsilon^{2^r}).
\end{eqnarray*}
Write $\beta=\mu\alpha+\upsilon$ and $\xi=\upsilon+\mu^{2^r}\tau+\upsilon^{2^r}$.
Then $\beta\in {\rm AGL}(\alpha), ~\xi\in \mathbb{F}_q$ and
$\beta^{2^r}=\beta+\xi.$
This gives
$\xi=\beta^{2^r}+\beta,$
which yields
$$\xi+\xi^{2^r}+\xi^{2^{2r}}+\cdots+\xi^{2^{(n-1)r}}=0.$$
It is known that $0,1,2,\cdots,n-1$ is a complete set of residues modulo $n$. From $\gcd(r,n)=1$ we have that
$0,r,2r,\cdots,(n-1)r$ is also a complete set of residues modulo $n$. By $\xi^{2^n}=\xi$ we have
$$\xi+\xi^{2}+\xi^{2^{2}}+\cdots+\xi^{2^{(n-1)}}=0.$$
This is equivalent to saying that (see \cite[Definition 2.22]{lidl})
${\rm Tr}_{\mathbb{F}_{q}/\mathbb{F}_{2}}(\xi)=0.$
Hence there exists $\omega_0\in \mathbb{F}_{q}$ such that $\xi=\omega_0^{2}+\omega_0$ (see \cite[Theorem 2.25]{lidl}).
Therefore, there is an element $\omega\in \mathbb{F}_{q}$ such that $\xi=\omega^{2^r}+\omega$,
because one can easily see that (using the fact $\gcd(r,n)=1$)
$$\big\{\eta^2+\eta\,|\,\eta\in \mathbb{F}_q\big\}=\big\{\eta^{2^r}+\eta\,|\,\eta\in \mathbb{F}_q\big\}.$$
It follows that
$$\beta^{2^r}=\beta+\xi=\beta+\omega^{2^r}+\omega,$$
which yields
$$(\beta+\omega)^{2^r}+(\beta+\omega)=0;$$
that is to say, $\beta+\omega\in {\rm AGL}(\alpha)$ is a root of $x^{2^r}+x$.
Let $g(x)$ be the minimal polynomial of $\beta+\omega$ over $\mathbb{F}_q$.
Then there is $g(x)\in {\rm PGL}(f)$ such that
$g(x)$ divides $x^{2^r}+x.$

(2) ${\rm AGL}\big(\sigma^r\big(\frac{1}{\alpha+\gamma}\big)\big)$=${\rm AGL}\big(\frac{1}{\alpha+\gamma}\big)$.
Using arguments essentially the same as those in case (1), we conclude that there exists $g(x)\in {\rm PGL}(f)$ such that
$g(x)$ divides $x^{2^r}+x.$

Conversely,
suppose that there is a polynomial $g(x)\in {\rm PGL}(f)$ such that  $g(x)$ divides $x^{2^r}+x$.
Then $f=D(g)$, where $D\in {\rm PGL}$. We then have  ${\rm PGL}(f)={\rm PGL}(g).$
In addition, from $f=D(g)$ we   can get that
$\sigma^r(f)=\sigma^rD(g)=\sigma^r(D)\sigma^r(g).$
Therefore
${\rm PGL}(\sigma^rf)={\rm PGL}(\sigma^rg).$
Let $\zeta$ be  a root of $g(x)$. Since $g(x)$ divides $x^{2^r}+x$,  we obtain  $\zeta^{2^r}=\zeta$, i.e.,
$\sigma^r(\zeta)=\zeta$, which implies that
${\rm PGL}(\zeta)={\rm PGL}\big(\sigma^r(\zeta)\big).$
Using Lemma \ref{equivalent},  one has
${\rm PGL}(g)={\rm PGL}(\sigma^rg).$
It follows that ${\rm PGL}(\sigma^rf)={\rm PGL}(f).$
The proof is complete.
\end{proof}
Now we are ready to determine the number of orbits of $\langle \sigma^r\rangle$ on ${\rm PGL}\verb|\|\mathcal{I}_r$ with size $1$.
Since $r\geq3$, we have that
$f(x)\in \mathcal{I}_r$ divides $x^{2^r}+x$  if and only if   $f(x)\in \mathcal{I}_r$ divides $x^{2^r-1}-1$ .
Let ${\rm ord}(f)$ denote the order of the polynomial $f$ (see \cite[Definition 3.2]{lidl}).
It follows from \cite[Lemma 3.6]{lidl} that $f(x)$ divides $x^{2^r}+x$ if and only if ${\rm ord}(f)$ divides $2^r-1$.
The set ${\rm E}(r,q)$ is defined by
\begin{equation}\label{erq}
{\rm E}(r,q)=\Big\{e \,\Big|\,\hbox{$e>1$ is an integer dividing $2^r-1$~but~$e$ does not divide $q^{d}-1$~for any~$1\leq d<r$}\Big\}.
\end{equation}
Then according to \cite[Theorem 3.5]{lidl},  the number of polynomials  $f(x)\in \mathcal{I}_r$
such that $f(x)$ divides $x^{2^r}+x$ is equal to
$$
\Big{|}\Big\{f(x)\in \mathcal{I}_r\,\Big|\,\hbox{$f(x)$ divides $x^{2^r}+x$}\Big\}\Big|=
\sum_{e\in {\rm E}(r,q)}\frac{\phi(e)}{r},$$
where $\phi$ is the Euler's function.

In the following we provide another characterization about the number of polynomials  $f(x)\in \mathcal{I}_r$
such that $f(x)$ divides $x^{2^r}+x$.

\begin{lem}\label{firstnumber}
With the notation as above. Then
$$\Big{|}\Big\{f(x)\in \mathcal{I}_r\,\Big|\,\hbox{$f(x)$ divides $x^{2^r}+x$}\Big\}\Big|=
\frac{1}{r}\sum_{d|r}\big(2^{\frac{r}{d}}-1\big)\mu(d),$$
where $\mu$ is the M$\ddot{o}$bius function.
\end{lem}
\begin{proof}
Let $d$ be a positive integer.
Suppose that
$$\Omega_d(x)=\prod_{f(x)\in \mathcal{S}_d(x)}f(x),$$
where
$$\mathcal{S}_d(x)=\Big\{f(x)\in \mathcal{I}_d\,\Big|\,\hbox{$f(x)$ divides $x^{2^r-1}-1$}\Big\}.$$
Note that $\gcd(2^r-1, q)=\gcd(2^r-1, 2^n)=1$ and $\gcd(r,n)=1$ and we get that
${\rm ord}_{2^r-1}(q)=r$, which shows that $d|r$. Then
$$x^{2^r-1}-1=\prod_{d|r}\Omega_d(x).$$
So
$$2^r-1=\sum_{d|r}d\big|\mathcal{S}_d(x)\big|.$$
In virtue of the M$\ddot{o}$bius inversion formula we obtain that
$$r\big|\mathcal{S}_r(x)\big|=\frac{1}{r}\sum_{d|r}\big(2^{\frac{r}{d}}-1\big)\mu(d).$$
Hence we have
$$\Big{|}\Big\{f(x)\in \mathcal{I}_r\,\Big|\,\hbox{$f(x)$ divides $x^{2^r}+x$}\Big\}\Big|=\big|\mathcal{S}_r(x)\big|=
\frac{1}{r}\sum_{d|r}\big(2^{\frac{r}{d}}-1\big)\mu(d),$$
where $\mu$ is the M$\ddot{o}$bius function.
\end{proof}

We have seen that  the number of monic irreducible polynomials of degree
$r$ over $\mathbb{F}_q$  that divide $x^{2^r}+x$ is equal to
$$\frac{1}{r}\sum_{e\in {\rm E}(r,q)}\phi(e) ~\mbox{or}~\frac{1}{r}\sum_{d|r}\big(2^{\frac{r}{d}}-1\big)\mu(d).$$


The following result reveals that if ${\rm PGL}(f)$ contains a polynomial that divides $x^{2^r}+x$, then
${\rm PGL}(f)$ contains exactly $6$ such polynomials.
\begin{lem}\label{lastlem}
Suppose that $f(x)\in \mathcal{I}_r$ such that $f(x)$ divides $x^{2^r}+x$. Then
$$\Big|\Big\{h(x)\,\Big|\,h(x)\in {\rm PGL}(f), ~\hbox{$h(x)$ divides $x^{2^r}+x$}\Big\}\Big|=6.$$
\end{lem}
\begin{proof}
For simplifying notations, let  $\Delta=\big\{h(x)\,\big|\,h(x)\in {\rm PGL}(f), ~\hbox{$h(x)$ divides $x^{2^r}+x$}\big\}$.
Let $\alpha$ be   a root of $f(x)$, which gives $\alpha^{2^r}=\alpha$ since $f(x)$ divides $x^{2^r}+x$.
Observe that
\begin{equation*}
\begin{split}
|\Delta|=&\Big|\Big\{h(x)\,\Big|\,h(x)\in {\rm PGL}(f), ~\hbox{$h(x)$ divides $x^{2^r}+x$}\Big\}\Big|\\
=&\Big|\Big\{Af(x)\,\Big|\,   \hbox{$Af(x)$ divides $x^{2^r}+x$}\Big\}\Big|\\
=&\Big|\Big\{A\in {\rm PGL}\,\Big|\,   \hbox{$Af(x)$ divides $x^{2^r}+x$}\Big\}\Big|\\
=&\Big|\Big\{A\in {\rm PGL}\,\Big|\,   (A\alpha)^{2^r}+A\alpha=0\Big\}\Big|.\\
\end{split}
\end{equation*}
Assume that $A=\begin{pmatrix} a & b \\c & d\end{pmatrix}.$
Then
\begin{eqnarray*}
(A\alpha)^{2^r}+A\alpha=0
&\Leftrightarrow & \Big(\frac{a\alpha+b}{c\alpha+d}\Big)^{2^r}+\Big(\frac{a\alpha+b}{c\alpha+d}\Big)=0\\
&\Leftrightarrow & \frac{a^{2^r}\alpha^{2^r}+b^{2^r}}{c^{2^r}\alpha^{2^r}+d^{2^r}}+\frac{a\alpha+b}{c\alpha+d}=0\\
&\Leftrightarrow & \frac{a^{2^r}\alpha+b^{2^r}}{c^{2^r}\alpha+d^{2^r}}+\frac{a\alpha+b}{c\alpha+d}=0\\
&\Leftrightarrow & (ca^{2^r}+ac^{2^r})\alpha^2+(da^{2^r}+bc^{2^r}+ad^{2^r}+cb^{2^r})\alpha+(bd^{2^r}+db^{2^r})=0\\
&\Leftrightarrow &
\begin{cases}
ca^{2^r}+ac^{2^r}=0,\\
da^{2^r}+bc^{2^r}+ad^{2^r}+cb^{2^r}=0,\\
bd^{2^r}+db^{2^r}=0.
\end{cases}
\end{eqnarray*}
Case 1: $a\neq 0, ~c=0$.
Since $A$ is invertible, $d\neq0$. From the second equality we have $a=d$.
If $b\neq 0$, then $b=d$. Hence, there are   two cases:
$$b=c=0,~a=d\neq0;~~c=0, ~a=b=d\neq0.$$
Therefore in this case
$$A=\begin{pmatrix} 1 & 0 \\0 & 1\end{pmatrix}~\mbox{or}~
\begin{pmatrix} 1 & 1 \\0 & 1\end{pmatrix}.$$
Case 2: $c\neq 0, ~a=0$.
Since $A$ is invertible, $b\neq0$ and $ c\neq 0$. By the second equality we obtain   $b=c$.
If $d\neq 0$, then $b=d$. Hence, there are   two cases:
$$b=c\neq 0,~a=d=0;~~a=0, ~b=c=d\neq0.$$
Therefore
$$A=\begin{pmatrix} 0 & 1 \\1 & 0\end{pmatrix}~\mbox{or}~
\begin{pmatrix} 0 & 1 \\1 & 1\end{pmatrix}.$$
Case 3: $c\neq 0, ~a\neq0$.
From the first equality we get $a=c$. We consider three subcases separately.

Subcase 3.1: $b=0, ~d\neq 0$. From the second equality we get $a=d$.

Subcase 3.2: $b\neq0, ~d=0$. From the second equality we get $b=c$.

Subcase 3.3: $b\neq0, ~d\neq0$. From the last equality we get $b=d$. However, the determinate of $A$ is  $ad-bc=0$. This is impossible.

Hence, there are  two cases:
$$b=0,~a=c=d\neq 0;~~d=0,~a=b=c\neq0.$$
Therefore in this case
$A=\begin{pmatrix} 1 & 0 \\1 & 1\end{pmatrix}~\mbox{or}~
\begin{pmatrix} 1 & 1 \\1 & 0\end{pmatrix}.$

In conclusion, we have
$\big|\Delta\big|=6.$
\end{proof}

\subsection{Our main result and two illustrative  examples}
Collecting all the results that we have established, we arrive at the following result, which gives an upper bound for the
number of inequivalent  extended irreducible binary Goppa codes of length $2^n+1$ and degree $r$.
\begin{Theorem}\label{theorem}
We assume that  $n>3$ is an odd prime number and
$q=2^n$; let $r\geq3$ be a positive integer satisfying $\gcd(r,n)=1$ and $\gcd\big(r,q(q^2-1)\big)=1$.
The number of inequivalent extended irreducible binary Goppa codes of length $q+1$ and degree $r$ is at most
$$\frac{n-1}{6rn}\sum_{d|r}\mu(d)\big(2^{\frac{r}{d}}-1\big)+\frac{1}{rnq(q^2-1)}\sum_{d\,|\,r}\mu(d)q^{\frac{r}{d}},
$$
where $\mu$ is the M\"{o}bius function.
\end{Theorem}
\begin{proof}
By Lemma \ref{important}, the number of inequivalent  extended irreducible binary Goppa codes of length $q+1$
and degree $r$ is less than or equal to the number of orbits of ${\rm P\Gamma L}$ on $\mathcal{S}$.
By Lemma \ref{orbit},
let $s$ be the number of orbits of ${\rm P\Gamma L}$ on $\mathcal{I}_r$. Using Lemmas \ref{nexttolastlem} and \ref{lastlem},
we have
$$\frac{1}{6r}\sum_{d|r}\mu(d)\big(2^{\frac{r}{d}}-1\big)+n\Big(s-\frac{1}{6r}\sum_{d|r}\mu(d)\big(2^{\frac{r}{d}}-1\big)\Big)
=\big|{\rm PGL}\verb|\|\mathcal{I}_r\big|=\frac{|\mathcal{I}_r|}{q(q^2-1)},$$
from which we obtain
$$s=\frac{n-1}{6rn}\sum_{d|r}\mu(d)\big(2^{\frac{r}{d}}-1\big)+\frac{|\mathcal{I}_r|}{nq(q^2-1)}.$$
Substituting  $|\mathcal{I}_r|=\frac{1}{r}\sum\limits_{d\,|\,r}\mu(d)q^{r/d}$ into the above equation, we obtain the desired result.
We are done.
\end{proof}

We give two small examples to illustrate Theorem \ref{theorem}.
\begin{Example}{\rm
Take $n=5$ and $r=7$ in Theorem \ref{theorem}. This gives $q=2^n=2^5=32$ and thus $q(q^2-1)=32(32^2-1)=31\cdot 32\cdot 33=32736$.
It is readily seen that
$\gcd(r,n)=\gcd(7,5)=1$ and $\gcd(r,q(q^2-1))=\gcd(7,32(32^2-1))=1$, namely, the conditions listed in Theorem \ref{theorem} are satisfied.
Then Theorem \ref{theorem} says that the number of inequivalent extended irreducible binary Goppa codes of length $33$ and degree $7$ is at most
\begin{eqnarray*}
&&\frac{n-1}{6rn}\sum_{d|r}\mu(d)\big(2^{\frac{r}{d}}-1\big)+\frac{1}{rnq(q^2-1)}\sum_{d\,|\,r}\mu(d)q^{\frac{r}{d}}\\
&=&\frac{4}{6\cdot7\cdot5}\sum_{d|7}\mu(d)\big(2^{\frac{7}{d}}-1\big)+\frac{1}{7\cdot5\cdot32(32^2-1)}\sum_{d\,|\,7}\mu(d)32^{\frac{7}{d}}\\
&=&\frac{12}{5}+\frac{1}{7\cdot5\cdot32(32^2-1)}(32^7-32)\\
&=&29991.
\end{eqnarray*}
}
\end{Example}

\begin{Example}{\rm
Take $n=7$ and $q=2^n=2^7=128$ and thus $q(q^2-1)=128(128^2-1)=2^7\cdot 3\cdot 43\cdot 127$.
Assume that
$\gcd(r,n)=\gcd(r,7)=1$ and $\gcd(r,q(q^2-1))=\gcd(r,2^7\cdot 3\cdot 43\cdot 127)=1$.
Then according to Theorem \ref{theorem}
the number of inequivalent extended irreducible binary Goppa codes of length $129$ and degree $r$ is at most
\begin{eqnarray*}
&&\frac{n-1}{6rn}\sum_{d|r}\mu(d)\big(2^{\frac{r}{d}}-1\big)+\frac{1}{rnq(q^2-1)}\sum_{d\,|\,r}\mu(d)q^{\frac{r}{d}}\\
&=&\frac{1}{7r}\sum_{d|r}\mu(d)\big(2^{\frac{r}{d}}-1\big)+\frac{1}{r\cdot 7\cdot2^7\cdot 3\cdot 43\cdot 127}\sum_{d\,|\,r}\mu(d)128^{\frac{r}{d}}.
\end{eqnarray*}
In the following we provide the values of the upper bounds on the number of inequivalent extended irreducible binary Goppa codes of length $129$ and some possible degrees $r$, which are listed in Table I.
\begin{table}\centering
\caption{Upper bounds on the number of inequivalent extended irreducible binary Goppa codes of length $129$ and degree $r$}
\begin{tabular}{c|l}  
\hline
\textbf{Degree}  & \textbf{Upper bound}\\
\hline
$r=~5$ & $469$ \\
\hline
$r=11$ & $935870030557051$\\
\hline
$r=13$ & $12974326183623782445$\\
\hline
$r=17$ & $2663294067654074513871726265$\\
\hline
$r=19$ & $39042208951344950852613887707059$\\
\hline
\end{tabular}
\end{table}

}
\end{Example}

\section{Concluding remarks and future work}
It is known that the number of inequivalent extended irreducible binary Goppa codes of length $2^n+1$
and degree $r$ is less than or equal to the number of orbits of ${\rm P\Gamma L}$ on $\mathcal{S}$.
In this paper, we present a new approach to
get an upper bound for the number of inequivalent extended irreducible binary Goppa codes by introducing a group action of
${\rm P\Gamma L}$ on $\mathcal{I}_r$, the set of all monic irreducible polynomials of degree $r$ over $\mathbb{F}_q$.
We show that the number of orbits of ${\rm P\Gamma L}$ on $\mathcal{S}$ is equal to the number of orbits of ${\rm P\Gamma L}$ on $\mathcal{I}_r$.
There are   some  advantages of considering the action of ${\rm P\Gamma L}$ on $\mathcal{I}_r$ which permits us to find a formula for the number of orbits of ${\rm P\Gamma L}$ on $\mathcal{S}$. Therefore, we obtain an upper bound for the number of inequivalent extended irreducible binary Goppa codes of length $2^n+1$ and degree $r$, where $n>3$ and $r$ satisfy certain conditions.

A possible
direction for future  work is to find tight   upper bounds for the number of inequivalent
extended irreducible binary Goppa codes in more cases.   It also would be interesting to find the exact value of extended irreducible binary Goppa codes.

\vspace{0.3 cm}
\noindent{\bf Acknowledgements}\quad
We are very grateful to
Professor Qin Yue in Nanjing University of Aeronautics and Astronautics for
introducing us to the topic of enumeration of Goppa codes, and sending their manuscript \cite{yueqin} to us.

\end{document}